\documentclass[sigconf]{acmart}

\acmYear{2026}
\acmConference[DAI '26]
  {The Eighth International Conference on Distributed Artificial Intelligence}
  {November 29--December 2, 2026}
  {Hong Kong, China}
\acmBooktitle{The Eighth International Conference on Distributed Artificial
  Intelligence (DAI '26), November 29--December 2, 2026, Hong Kong, China}

\setcopyright{none}
\copyrightyear{}
\acmDOI{}
\acmISBN{}

\AtBeginDocument{%
  }

\usepackage[utf8]{inputenc} 
\usepackage[T1]{fontenc}    
\usepackage{hyperref}       
\usepackage{url}            
\usepackage{booktabs}       
\usepackage{amsfonts}       
\usepackage{nicefrac}       
\usepackage{microtype}      
\usepackage{xcolor}         
\usepackage{wrapfig}
\usepackage{graphicx}
\usepackage{subfigure}
\usepackage{amsmath}
\usepackage{mathtools}
\usepackage{amsthm}

\theoremstyle{plain}

\newtheorem{Theorem}{Theorem}
\newtheorem{Definition}{Definition}
\newtheorem{Corollary}{Corollary}
\newtheorem{Proposition}{Proposition}
\newcommand{\nop}[1]{{}}

\usepackage[textsize=tiny]{todonotes}

\begin{document}

\title{Co-opetition Equilibrium in Adversarial Team Games with Heterogeneous Utilities}
\author{Youzhi Zhang}
\affiliation{
  \institution{Centre for Artificial Intelligence and Robotics, Hong Kong Institute of Science \& Innovation, Chinese Academy of Sciences}
  \city{Hong Kong SAR}
  \country{China}}
\email{youzhi.zhang@cair-cas.org.hk}




\begin{abstract}
The United Nations' 2030 Agenda for Sustainable Development requires all countries to collaborate against adversarial factors, a scenario that can be formalized as an adversarial team game. However, existing solution concepts assume team players share identical utility functions---an assumption inconsistent with real-world settings where countries have divergent objectives. This paper argues that studying adversarial team games must account for heterogeneous utilities among team players. We show that ignoring utility differences can render computed equilibria unstable in the original game, and we formalize this degradation via the \emph{Price of Ignoring Heterogeneity}, which can be unbounded. To address this, we introduce the Co-opetition Equilibrium (CoE), where team players with heterogeneous utilities correlate their strategies (cooperation) against an adversary (competition). We establish existence via reduction from Nash equilibrium, and prove that finding a CoE is PPAD-complete while computing a Team-Maximizing CoE (TMCoE) is NP-hard. Nevertheless, we identify a broad class of zero-sum adversarial team games---those satisfying a consistent-constraint condition that generalizes identical utilities---where TMCoEs are exchangeable and computable in polynomial time via linear programming.   We outline directions for algorithm design, MARL, and extensions to extensive-form and multi-team games.
\end{abstract}

\keywords{Co-opetition; Cooperation; Competition; Team-maxmin equilibrium}

\ccsdesc[500]{Theory of computation~Algorithmic game theory}


\maketitle 

\section{Introduction}
In 2015, the United Nations (UN) Member States adopted the 2030 Agenda for Sustainable Development~\cite{UN2023}, calling for all countries to collaborate against adversarial factors to achieve peace and prosperity. In AI, this scenario can be formalized as an \emph{adversarial team game}~\cite{von1997team,basilico2017team,celli2018computational,zhang2020computing,farina2021connecting,zhang2023team,liu2024computing,zhang2024dag}, where a team of players competes against an adversary (or a team of opponents).

Prior work on adversarial team games~\cite{von1997team,farina2018ex,zhang2020converging,carminati2022marriage,zhang2022correlation,zhang2022teamtree,zhang2022subgame,mcaleer2023teampsro,liu2024leveraging} rests on two assumptions: (i)~team players share the same utility function, and (ii)~the game is zero-sum. Under these assumptions, several solution concepts exist. When team players act independently, the relevant concepts are Nash Equilibrium (NE)~\cite{nash1951non,Anagnostides2023Algorithms} and Team-Maxmin Equilibrium (TME)~\cite{von1997team,basilico2017team,celli2018computational,zhang2020converging}, where the latter maximizes the team's utility among all NEs. When team players can correlate their strategies~\cite{basilico2017team}, the Correlated Team-Maxmin Equilibrium (CTME) reduces to a two-player zero-sum NE, because identical utilities allow the team to be treated as a single player. In extensive-form games, two variants exist~\cite{celli2018computational,farina2018ex,zhang2021computing,zhang2022optimal,zhang2023team}: CTME with a communication device (team can always communicate; equivalent to perfect-recall two-player zero-sum) and CTME with a coordination device (team communicates only before play; equivalent to imperfect-recall two-player zero-sum).

However, these solution concepts fail when team players have heterogeneous utilities---a common real-world scenario. For instance, at the 2022 UN Biodiversity Conference, the Kunming-Montreal Global Biodiversity Framework was accepted by most countries but rejected by the Democratic Republic of Congo due to unmet funding requirements~\cite{Guardian2022}. Similarly, at the 2023 UN Climate Change Conference, oil-producing countries opposed fossil-fuel phase-out policies that most other nations supported~\cite{CNN2023}. These examples illustrate that countries do not always share the same objectives,  necessitating new solution concepts for adversarial team games with heterogeneous utilities.

\textbf{Contributions.} This paper makes the following contributions:
\begin{enumerate}
    \item We demonstrate that ignoring utility differences among team players can produce equilibria that are unstable in the original game. We formalize this gap via the \emph{Price of Ignoring Heterogeneity} (PoIH), and prove that it can be unbounded.
    \item We introduce the \emph{Co-opetition Equilibrium} (CoE), a novel solution concept for adversarial team games with heterogeneous team utilities, where team players correlate their actions (cooperation) while competing against an adversary (competition). We further define the \emph{Team-Maximizing CoE} (TMCoE), which maximizes the team's utility among all CoEs.
    \item We give an exact condition for when a CoE induces a CE of the full game. We establish that each NE induces a CoE, guaranteeing existence in all finite adversarial team games. We prove that finding a CoE is PPAD-complete   and computing a TMCoE is NP-hard, via a non-trivial reduction from two-player games.
    \item We identify a broad class of zero-sum adversarial team games satisfying \emph{consistent constraints}---a condition that generalizes identical utilities---where TMCoEs are exchangeable and computable in polynomial time via linear programming. We further analyze relaxations of this condition, including $\epsilon$-approximate consistency and ordinal consistency.
    \item We discuss connections to concave games, coalitional games, and MARL, and outline promising directions for future work.
\end{enumerate}

\section{Preliminaries}
A normal-form game $G$~\cite{shoham2008multiagent} consists of a set $N=\{1,\dots,n\}$ of players, a finite action set $A_i$ for each player $i\in N$, and a utility function $u_i:\mathbf{A}\rightarrow \mathbb{R}$ for each player $i\in N$, where $\mathbf{A}=\times_{i\in N}A_i$ is the set of joint actions and $a_i\in A_i$, $\mathbf{a}\in \mathbf{A}$.

The set of player $i$'s mixed strategies is $X_i =\Delta(A_i)$; each $x_i\in X_i$ is a probability distribution over $A_i$. $\mathbf{X}=\times_{i\in N}X_i$ is the set of mixed strategy profiles. Player $i$'s expected utility for $\mathbf{x}\in \mathbf{X}$ is $u_i(\mathbf{x})=\sum_{\mathbf{a}\in \mathbf{A}}u_i(\mathbf{a})\prod_{j\in N, a_j\in \mathbf{a}}x_j(a_j)$. We use $-i=N\setminus\{i\}$ to denote all players except $i$, with $\mathbf{A}_{-i}=\times_{j\in -i}A_j$, $\mathbf{X}_{-i}=\times_{j\in -i}X_j$, and $u_i(a_i,\mathbf{x}_{-i})=\sum_{\mathbf{a}_{-i}\in \mathbf{A}_{-i}}u_i(a_i,\mathbf{a}_{-i})\prod_{j\in-i,a_j\in\mathbf{a}_{-i}}x_j(a_j)$. For $N'\subseteq N$, $\mathbf{A}_{N'}=\times_{j\in N'}A_j$.

A strategy profile $\mathbf{x}=(x_i,\mathbf{x}_{-i})\in\mathbf{X}$ is a \textbf{Nash Equilibrium (NE)} if no player can improve her expected utility by unilaterally deviating: $u_i(x_i,\mathbf{x}_{-i})\geq u_i(x'_i,\mathbf{x}_{-i})$ for all $i\in N$ and $x'_i\in X_i$.
A correlated strategy $x_{N'}\in X_{N'}=\Delta(\mathbf{A}_{N'})$ is a probability distribution over the joint action space of players in $N'$. We denote $x_{N}=x$. Note that we use bold symbols ($\mathbf{x}_T$, $\mathbf{x}_{-i}$, \ldots) exclusively for profiles of independent mixed strategies and non-bold symbols ($x_T$, $x$, \ldots) for correlated distributions. 
$x\in X = \Delta(\mathbf{A})$ is a \textbf{Correlated Equilibrium (CE)}~\cite{aumann1974subjectivity} if, when a mediator draws $\mathbf{a}\in \mathbf{A}$ according to $x$ and privately recommends $a_i$ to each player $i$, no player benefits from deviating from the recommendation: $\forall i\in N, a_i,a'_i\in A_i,$
\begin{equation} 
\sum_{\mathbf{a}_{-i}\in \mathbf{A}_{-i}}x(a_i,\mathbf{a}_{-i})(u_i(a_i,\mathbf{a}_{-i})-u_i(a'_i,\mathbf{a}_{-i}))\geq 0.
\end{equation}

In this paper, we consider \textbf{adversarial team games}~\cite{von1997team,basilico2017team,zhang2020converging}, where a team $T=N\setminus\{n\}$ coordinates against an adversary (player $n$). Let $u_T(\mathbf{x})=\sum_{i\in T}u_i(\mathbf{x})$ denote the team's expected utility under profile $\mathbf{x}=(\mathbf{x}_T,x_n)$ with $\mathbf{x}_T\in\mathbf{X}_T$ and $x_n\in X_n$. In zero-sum games with identical team utilities, $\mathbf{x}^\star=(\mathbf{x}_T,x_n)$ is a \textbf{Team Maxmin Equilibrium (TME)} if it is an NE that maximizes the team's utility. The \textbf{Correlated Team Maxmin Equilibrium (CTME)} is equivalent to a two-player NE because identical utilities allow the team to be treated as a single player.

\section{The Cost of Ignoring Utility Differences Among Team Players}\label{section:cost}
 
Marschak's seminal work on team theory~\cite{marschak1955elements} argues that utility differences among team members should be ignored to simplify analysis. Following this tradition, prior work on adversarial team games~\cite{von1997team,basilico2017team,celli2018computational,zhang2020computing,farina2021connecting,zhang2023team} assumes team players share a common utility. However, we show that this simplification can cause computed equilibria to be unstable in the original game, incurring significant losses.

\begin{table}[h]
    \begin{minipage}[b]{0.42\textwidth}\centering
        \begin{tabular}{c|cc}
       $G_a$  &D&C  \\\hline
      D   &0,0&7,2 \\
      C& 2,7&6,6
    \end{tabular}
    \caption{ A game of chicken}
      \label{tab:chickengame}
    \end{minipage}
    \begin{minipage}[b]{0.43\textwidth}\centering
        \begin{tabular}{c|cc}
       $G_b$  &D&C  \\\hline
      D   &0,0&4.5,4.5 \\
      C& 4.5,4.5&6,6
    \end{tabular}
    \caption{ Ignoring the difference}
      \label{tab:chickengameignore}
    \end{minipage}
\end{table}

We illustrate this issue with an example extended from the classic game of chicken~\cite{wiki2025ce}. The team has two players, and the adversary has two actions $A$ and $B$, where $A$ is strictly dominated by $B$. When the adversary plays $B$, the team's payoff matrix is the game of chicken (Table~\ref{tab:chickengame}), where D is ``Dare'' and C is ``Chicken out.'' Let $G_a$ denote the original game (Table~\ref{tab:chickengame}) and $G_b$ denote the game obtained by ignoring utility differences (Table~\ref{tab:chickengameignore}), in both cases given adversary action $B$.

In the original game of chicken, $(D,C)$, $(C,D)$, and $((1/3\,D,\,2/3\,C)$, $(1/3\,D,\,2/3\,C))$ are three NEs, yielding team utilities of $9$, $9$, and $84/9$ respectively. In normal-form adversarial team games, TME and CTME require identical team utilities. Applying these concepts to $G_a$ requires ignoring the utility differences, producing $G_b$, where $(C,C,B)$ is the unique NE and TME (since $D$ is strictly dominated by $C$ for each player). However, $(C,C,B)$ is \emph{not} an NE in the original game $G_a$: given $(C,C,B)$, each team player would deviate to $D$ to obtain utility $7$ instead of $6$, resulting in the outcome $(D,D)$ with utility $0$ for both team players. Similarly, $((C,C),B)$ is a CTME in $G_b$ but unstable in $G_a$: if $C$ is recommended, each team player deviates to $D$, leading to $(D,D)$ and utility $0$ for both team players. 
Overall, given $(C,C,B)$, each team player has a profitable unilateral deviation to $D$, obtaining utility 7 instead of 6. After one such deviation the realized outcome is $(D,C)$ or $(C,D)$ (team utility $9$); if both players simultaneously deviate the realized outcome is $(D,D)$ (team utility $0$). In either case the original profile is unstable.

Thus, ignoring utility differences can destabilize computed equilibria and cause severe losses for the team. 

\begin{table*}[!tp]
    \centering
        \begin{tabular}{l|cccccccc}
     &$G_a$  &$G_a$   &$G_a$  &$G_b$ &$G_b$  & $G_b$      &$G_a$ &$G_a$  \\\hline
    Strategy Profile &  $u_T$ &Stable   &NE  &NE &TME  & CTME      &CoE &TMCoE  \\\hline
    1  (C, C, B) for three players&0    &no & &yes  &  yes&&&\\\hline
    2  ((C, C), B) for team $T$ and player $n$&0    &no &  &&&   yes&&\\\hline
    3  (D, C, B) & 9  &yes &yes  &&&  & & \\\hline
    4  ((D, C), B)  & 9  &yes &&&&  & yes& \\\hline
    5  (C, D, B) &9  &yes& yes  &&  && \\\hline
    6  ((C, D), B)& 9  &yes &&&& &yes&  \\\hline
    7   ((1/3, 2/3), (1/3, 2/3), B) &$\frac{84}{9}$&yes& yes& &&&& \\\hline
    8   ((4/9, 2/9, 2/9, 1/9), B)  & $\frac{84}{9}$ &yes&  &&  && yes& \\\hline
    9  ((1/3, 1/3, 1/3, 0), B)  &10 &yes &&&&& yes&  \\\hline
    10  ((1/2, 1/4, 1/4, 0), B)  &$\frac{21}{2}$& yes & &  &&& yes& yes \\
    \hline
    \end{tabular}
    \caption{%
     $G_a$ and $G_b$ represent adversarial team games with payoff matrices from Tables~\ref{tab:chickengame} and~\ref{tab:chickengameignore} given adversary action $B$. Profiles 7--10 are: ((1/3 D, 2/3 C), (1/3 D, 2/3 C), B); ((4/9 (C,C), 2/9 (D,C), 2/9 (C,D), 1/9 (D,D)), B); ((1/3 (C,C), 1/3 (D,C), 1/3 (C,D), 0 (D,D)), B); and ((1/2 (C,C), 1/4 (D,C), 1/4 (C,D), 0 (D,D)), B). CoEs 4, 6, and 8 are induced by NEs 3, 5, and 7 respectively. The TME and CTME in $G_b$ are unstable in $G_a$ (see Section~\ref{section:cost}). To make the table clear, we use empty cells to represent `no' in the table.}
    
    \label{tab:differentequilibria}
\end{table*}


\section{Co-opetition Equilibrium}\label{sec_coopetition}
 
To address heterogeneous team utilities, we introduce the \emph{Co-opetition Equilibrium} (CoE). In an adversarial team game, a stable state should satisfy two conditions: (i)~the adversary's strategy is a best response to the team's correlated strategy (competition), and (ii)~team players have no incentive to deviate from the mediator's recommended actions given the adversary's strategy (cooperation). Formally:

\begin{Definition}\label{def:CoE}
 $(x_T,x_n)$ with $x_T\in X_T$ and $x_n\in X_n$ is a \textbf{Co-opetition Equilibrium (CoE)} in an adversarial team game if:
\begin{enumerate}
    \item  $x_n$ is a best response to $x_T$: $u_n(x_T,x_n)\geq u_n(x_T,x'_n)$ for all $x'_n\in X_n$;
    \item  $x_T$ satisfies the CE constraints for the team given $x_n$: for all $i\in T$ and $a_i,a'_i\in A_i$, with $-\{i,n\}=N\setminus\{i,n\}$,
    \begin{equation}\label{GCTMEConstraint}
          \sum_{\mathbf{a}_{-(i,n)}\in \mathbf{A}_{-(i,n)}} \!\!\!\!\!\!\!x_T(a_i,\mathbf{a}_{-(i,n)})\left(u_i(a_i,\mathbf{a}_{-(i,n)},x_n)
      -u_i(a'_i,\mathbf{a}_{-(i,n)},x_n)\right)\geq 0.
    \end{equation}
\end{enumerate}
\end{Definition}
A CoE is realized under the following  modeling assumptions:

\textbf{(i) Internal correlation device.}
The team has access to a mediator (or an equivalent pre-play communication protocol) that privately recommends actions drawn from an arbitrary joint distribution $x_T\in\Delta(A_T)$. Each team member commits to follow the recommendation whenever it is incentive-compatible, i.e., whenever the CE constraints in Eq. (\ref{GCTMEConstraint}) are satisfied.
In the setting of the United Nations’ 2030 Agenda for Sustainable Development, the UN itself (or a designated coordinating body such as a COP presidency) can serve as this mediator: it circulates privately negotiated package deals that function as correlated recommendations to the participating countries.

\textbf{(ii) External independence.}
No correlation is possible between the team’s joint action and the adversary’s action.
Continuing the UN example, the adversarial factors—those that threaten peace and prosperity for people and the planet, such as poverty and other deprivations, inequality, climate change, and related systemic risks—lie outside the coordinating body and therefore cannot be forced to correlate with the team’s recommendations.


\subsection{Example and Advantages of CoE}
 
We illustrate CoE using the game $G_a$ from Section~\ref{section:cost}. The following strategy profiles are CoEs in $G_a$:
\begin{itemize}
    \item $((D, C), B)$, $((C, D), B)$, and $(((1/3\,D,\,2/3\,C), (1/3\,D,\,2/3\,C)),$ $ B)$, induced by the three NEs in the original game of chicken, with team utilities $9$, $9$, and $84/9$.
    \item $((1/3\,(C,C),\,1/3\,(D,C),\,1/3\,(C,D),\,0\,(D,D)), B)$ is a CoE with team utility $10$, since the team's correlated strategy is a CE in the original game of chicken.
    \item $((1/2\,(C,C),\,1/4\,(D,C),\,1/4\,(C,D),\,0\,(D,D)), B)$ is a CoE  with maximum team utility $21/2$.
\end{itemize}
Compared to the team's utility of $0$ under TME and CTME (due to instability), CoE achieves utilities up to $21/2$, higher than the best NE utility of $84/9$. Table~\ref{tab:differentequilibria} summarizes these comparisons. 

The examples above show that ignoring utility differences can be catastrophic: the team's utility could drop from $21/2$ (CoE) to $0$ (under TME/CTME, due to deviation). We formalize this degradation via the \emph{Price of Ignoring Heterogeneity} (PoIH).

\begin{Definition}\label{def:PoIH}
Let \(G\) be an adversarial team game and let \(\bar{G}\) be the homogenized game in which every team player \(i\in T\) receives the averaged utility \(\bar{u}_i(\mathbf{a})=u_T(\mathbf{a})/|T|\).
For a strategy profile \(\sigma\) of \(\bar{G}\), write \(\sigma|_G\) for the same profile interpreted in the original game \(G\).
Define the realized team utility of \(\sigma\) in \(G\) by
\[
u_T^{\mathrm{real}}(\sigma)
\;=\;
\begin{cases}
u_T(\sigma|_G) & \text{if }\sigma|_G\text{ is one of $G$'s  Nash equilibria}\\
& \text{(respectively,  correlated equilibria)},\\[4pt]
u_T\bigl(\mathrm{Dev}(\sigma)\bigr) & \text{otherwise},
\end{cases}
\]
where the deviation process \(\mathrm{Dev}(\sigma)\) is:
  Team members who have a profitable unilateral deviation (or a profitable deviation from the recommended action, in the correlated case) simultaneously switch to a pure best response against the remaining strategy profile; if several best responses exist, any fixed tie-breaking rule may be used.
  The resulting (possibly non-equilibrium) outcome is denoted \(\mathrm{Dev}(\sigma)\), and the team utility actually obtained is \(u_T(\mathrm{Dev}(\sigma))\).
Let
\begin{align*}
V_{\mathrm{hom}}^{\mathrm{TME}}(G)
  &\;=\;
  \max\bigl\{\,u_T^{\mathrm{real}}(\sigma)
  \;\big|\;
  \sigma\text{ is a TME of }\bar{G}\,\bigr\},\\
  V_{\mathrm{hom}}^{\mathrm{CTME}}(G)
  &\;=\;
  \max\bigl\{\,u_T^{\mathrm{real}}(\sigma)
  \;\big|\;
  \sigma\text{ is a CTME of }\bar{G}\,\bigr\},\\
  V_{\mathrm{CoE}}(G) &\;=\;\max\{u_T(x_T,x_n)\mid(x_T,x_n)\text{ is a CoE of }G\}.
\end{align*}
The two \textbf{Prices of Ignoring Heterogeneity (PoIH)} are:
\begin{align*}
\mathrm{PoIH}_{\mathrm{TME}}(G)
  &=
  V_{\mathrm{CoE}}(G)-V_{\mathrm{hom}}^{\mathrm{TME}}(G), \\
\mathrm{PoIH}_{\mathrm{CTME}}(G)
&=
  V_{\mathrm{CoE}}(G)-V_{\mathrm{hom}}^{\mathrm{CTME}}(G).
\end{align*}
 When the distinction between the two classical solution concepts is immaterial we simply write \(\mathrm{PoIH}(G)\) for either quantity.
\end{Definition}

PoIH quantifies how much team utility is lost by applying traditional solution concepts (which assume identical utilities) to a game with heterogeneous utilities. 
The following result shows that this loss can be arbitrarily large.

\begin{Theorem}\label{thm:PoIH}
There exists a family of adversarial team games $\{G_m\}_{m=1}^\infty$ such that, as $m\to\infty$,
\[
\mathrm{PoIH}_{\mathrm{TME}}(G_m)\to+\infty
\quad\text{and}\quad
\mathrm{PoIH}_{\mathrm{CTME}}(G_m)\to+\infty.
\]
\end{Theorem}
\begin{proof}
The game \(G_a\) of Table~\ref{tab:chickengame} already yields \(\mathrm{PoIH}(G_a)=+\infty\): the unique TME/CTME of the homogenized game is \((C,C,B)\), which is unstable in \(G_a\) and, after the deviation process, produces team utility \(V_{\mathrm{hom}}=0\), while a CoE of \(G_a\) attains team utility \(21/2\).
A parameterized family with unbounded finite PoIH is obtained by scaling the off-diagonal payoffs in   the chicken game. For any \(m>5\), let $u_1(D,C)=u_2(C,D)=m, 
u_1(C,D)=u_2(D,C)=m-5, 
u_1(C,C)=u_2(C,C)=m-1, 
u_1(D,D)=u_2(D,D)=0.$
In the natural adversarial-team extension of the chicken game (analogous to the example in Section \ref{section:cost}), the homogenized TME/CTME again collapses, after deviation, to team utility \(0\), while a CoE extracts team utility \(\Theta(m)\). Hence
$
\mathrm{PoIH}(G_m)=\Theta(m)\to+\infty
$  as $m\to\infty$.
\end{proof}

\subsection{Relationship to CE of the Full Game} 
A CoE $(x_T, x_n)$ induces a joint distribution $x \in \Delta(\mathbf{A})$ via $x(\mathbf{a}_T, a_n) = x_T(\mathbf{a}_T) \cdot x_n(a_n)$. A natural question is how this induced distribution relates to the CE of the full $n$-player game. The following proposition shows that every CoE induces a CE, but the converse requires independence between team and adversary actions.

\begin{Proposition}\label{prop:CoEvsCE}
 A CoE $(x_T, x_n)$ induces a CE $x = x_T \otimes x_n$ of the full game. Conversely, a CE $x \in \Delta(\mathbf{A})$ of the full game corresponds to a CoE if and only if $x$ factorizes as $x(\mathbf{a}_T, a_n) = x_T(\mathbf{a}_T) \cdot x_n(a_n)$ for some $x_T \in X_T$, $x_n \in X_n$.
\end{Proposition}
\begin{proof}
($\Rightarrow$) Let $(x_T, x_n)$ be a CoE. For the induced product distribution $x = x_T \otimes x_n$, the CE constraint for each team member $i \in T$ reduces directly to Eq.~\eqref{GCTMEConstraint}, which holds by Definition~\ref{def:CoE}. For the adversary $n$, the CE constraint requires $\sum_{\mathbf{a}_T} x_T(\mathbf{a}_T) x_n(a_n)(u_n(\mathbf{a}_T, a_n) - u_n(\mathbf{a}_T, a'_n)) \geq 0$ for all $a_n, a'_n \in A_n$. When $x_n(a_n) > 0$, this simplifies to $u_n(x_T, a_n) \geq u_n(x_T, a'_n)$, which holds because $x_n$ is a best response to $x_T$ and $a_n$ is in its support. Thus $x$ is a CE of the full game.

($\Leftarrow$) Let $x = x_T \otimes x_n$ be a CE of the full game. The team CE constraints for $i \in T$ simplify to Eq.~\eqref{GCTMEConstraint}, so $x_T$ is a CE for the team given $x_n$. The adversary CE constraint gives $\sum_{\mathbf{a}_T} x_T(\mathbf{a}_T)(u_n(\mathbf{a}_T, a_n) - u_n(\mathbf{a}_T, a'_n)) \geq 0$ whenever $x_n(a_n) > 0$, i.e., $u_n(x_T, a_n) \geq u_n(x_T, a'_n)$ for all $a_n \in \text{support}(x_n)$ and $a'_n \in A_n$. Taking the convex combination over $a_n$ with weights $x_n(a_n)$ yields $u_n(x_T, x_n) \geq u_n(x_T, x'_n)$ for all mixed $x'_n$. Hence $x_n$ is a best response to $x_T$, and $(x_T, x_n)$ is a CoE.  
\end{proof}

As shown in Proposition~\ref{prop:CoEvsCE}, every CoE induces a CE of the full game (specifically, a product CE where team and adversary actions are independent). However, not every CE is a CoE: a CE may involve correlation between team and adversary actions, or the adversary may not be best-responding to the team's marginal strategy. Intuitively, CoE refines CE for adversarial settings by  imposing independence between team and adversary. The constraints in Eq.~\eqref{GCTMEConstraint} are bilinear in $(x_T, x_n)$, reflecting this coupling---whereas standard CE constraints are linear in the joint distribution $x$. The following example illustrates a CE that is \emph{not} a CoE---and shows why enforcing CE in adversarial settings can be disastrous when the adversary deviates from the mediator.


 We extend the game of chicken~\cite{wiki2025ce}. There are three players with two actions each. Player~1 (actions $a_1,a_2$) and Player~2 (actions $b_1,b_2$) form the team; Player~3 (actions $c_1,c_2$) is the adversary. Table~\ref{tab:modifiedgameofchicken} shows utilities when Player~1 plays $a_1$; all other action profiles yield utility $0$ for all players.

\begin{table}[h]
    \centering
\begin{tabular}{c|cc}
       $a_1$  &$c_1$&$c_2$  \\\hline
      $b_1$   &0,0,0&0,7,2 \\
      $b_2$& 0,2,7&0,6,6
    \end{tabular}
    \caption{Utilities when Player~1 plays $a_1$.}
    \label{tab:modifiedgameofchicken}
\end{table}

The CoEs are: $((a_1,b_1),c_2)$ with team utility $7$; $((a_1,b_2),c_1)$ with team utility $2$; and $((1/3\,(a_1,b_1),\,2/3\,(a_1,b_2)), (1/3\,c_1,\,2/3\,c_2))$ with team utility $42/9$. 

Now consider a CE of the full game that is \emph{not} a CoE: the joint distribution $(1/2\,(a_1,b_2,c_2),\,1/4\,(a_1,b_1,c_2),\,1/4\,(a_1,b_2,c_1))$. This is a CE because, conditional on each recommendation, no player has a profitable deviation. However, it fails to be a CoE for two reasons: (i)~the team and adversary actions are correlated (the adversary's recommended action depends on the team's action profile), violating the product form; and (ii)~the adversary's marginal $x_n = (3/4\,c_2,\,1/4\,c_1)$ is \emph{not} a best response to the team's marginal $x_T = (1/4\,(a_1,b_1),\,3/4\,(a_1,b_2))$---the adversary's true best response to $x_T$ is $c_1$ (yielding utility $21/4$ vs.\ $20/4$ for $c_2$). If the adversary deviates from the CE recommendation and best-responds, the team utility drops to $3/2$, compared to $17/4$ if the CE were followed. This example demonstrates that CE without the CoE structure is unreliable in adversarial settings: the adversary has no incentive to follow a mediator, and a CE that recommends adversary actions is not self-enforcing against a strategic opponent.

\textbf{Remark (CCE variant).} Our framework extends naturally to Coarse Correlated Equilibrium (CCE)~\cite{moulin1978strategically}: replacing the CE constraints~\eqref{GCTMEConstraint} with the weaker CCE constraints (which only require no ex ante regret, without conditioning on the recommended action) yields a \emph{Coarse CoE}, with a larger equilibrium set and potentially cheaper approximation, at the cost of weaker incentive guarantees for team members. All structural results  carry over with CCE constraints.

\section{Theoretical Properties of CoE}\label{sec_theory}
 
We now study the existence and computational complexity of CoE.

\subsection{Existence via Reduction from NE}
We first establish that every NE in an adversarial team game induces a CoE, which immediately guarantees existence.

\begin{Theorem}\label{theorem_nash_correlated_nash}If $(\mathbf{x}_T,x_n)$ with $\mathbf{x}_T\in\mathbf{X}_T$ and $x_n\in X_n$ is an NE in an adversarial team game, then $(x_T,x_n)$ where $x_T(\mathbf{a}_T)=\prod_{i\in T, a_i\in\mathbf{a}_T}x_i(a_i)$ for each $\mathbf{a}_T\in\mathbf{A}_T$ is a CoE.
\end{Theorem}
\begin{proof}
By the definition of NE, $x_n$ is a best response to $x_T$. For each $i\in T$, $a_i\in A_i$, $x_i\in \mathbf{x}_T$, we consider:

\textbf{Case 1:} If $x_i(a_i)=0$, then for each $a'_i\in A_i$:
\begin{equation*}
\begin{split}
    &0=x_i(a_i)u_i(a_i,\mathbf{x}_{-i})\geq  x_i(a_i)u_i(a'_i,\mathbf{x}_{-i})=0\\
  \Rightarrow &x_i(a_i)\sum_{\mathbf{a}_{-(i,n)}\in \mathbf{A}_{-(i,n)}}\prod_{j\in -(i,n)} x_j(a_j)u_i(a_i,\mathbf{a}_{-(i,n)},x_n) \\
  &\geq  x_i(a_i)\sum_{\mathbf{a}_{-(i,n)}\in \mathbf{A}_{-(i,n)}}\prod_{j\in -(i,n)} x_j(a_j)u_i(a'_i,\mathbf{a}_{-(i,n)},x_n)\\
  \Rightarrow & \sum_{\mathbf{a}_{-(i,n)}\in \mathbf{A}_{-(i,n)}}x_T(a_i,\mathbf{a}_{-(i,n)})u_i(a_i,\mathbf{a}_{-(i,n)},x_n) \\
  &\geq \sum_{\mathbf{a}_{-(i,n)}\in \mathbf{A}_{-(i,n)}}x_T(a_i,\mathbf{a}_{-(i,n)})u_i(a'_i,\mathbf{a}_{-(i,n)},x_n).
\end{split}
\end{equation*}

\textbf{Case 2:} If $x_i(a_i)>0$, by the NE property~\cite{sandholm2005mixed}, for each $a'_i\in A_i$:
\begin{equation*}
\begin{split}
& u_i(\mathbf{x})=u_i(a_i,\mathbf{x}_{-i})\geq u_i(a'_i,\mathbf{x}_{-i})\\
\Rightarrow &x_i(a_i)u_i(a_i,\mathbf{x}_{-i})\geq x_i(a_i)u_i(a'_i,\mathbf{x}_{-i})\\
\Rightarrow & \sum_{\mathbf{a}_{-(i,n)}\in \mathbf{A}_{-(i,n)}}x_T(a_i,\mathbf{a}_{-(i,n)})u_i(a_i,\mathbf{a}_{-(i,n)},x_n)\\
&\geq \sum_{\mathbf{a}_{-(i,n)}\in \mathbf{A}_{-(i,n)}}x_T(a_i,\mathbf{a}_{-(i,n)})u_i(a'_i,\mathbf{a}_{-(i,n)},x_n).
\end{split}
\end{equation*}

Therefore, given $x_n$, $x_T$ satisfies the CE constraints~\eqref{GCTMEConstraint}, so $(x_T,x_n)$ is a CoE.
\end{proof}
\begin{Corollary}\label{cor:existence}
 A CoE always exists in any finite adversarial team game.
\end{Corollary}
\begin{proof}
An NE always exists in any finite game~\cite{nash1951non}. The result follows from Theorem~\ref{theorem_nash_correlated_nash}.
\end{proof}

\subsection{Computational Complexity of CoE}
We now characterize the complexity of finding a CoE. The key idea is to reduce any two-player game to a three-player zero-sum adversarial team game, establishing a bidirectional correspondence between Nash equilibria and CoEs; this yields PPAD-hardness. Membership in PPAD follows from a polynomial-time reduction that maps an arbitrary CoE instance to a standard PPAD problem. Together the two directions establish that finding a CoE is PPAD-complete.

\begin{Theorem}\label{reductionfromzpto30p}
For any two-player game $G_2=(\{2,3\},A_2\times A_3,(u_2,u_3))$, there exists a three-player zero-sum adversarial team game $G_T=(\{1,2,3\},A_1\times A_2\times A_3,(u_1,u_2,u_3))$ with $T=\{1,2\}$ and $A_1=\{a_1\}$ such that $(x_2,x_3)$ is an NE in $G_2$ if and only if $(x_T,x_3)$ with $x_T(a_1,a_2)=x_2(a_2)$ is a CoE in $G_T$.
\end{Theorem}
\begin{proof}
For any two-player game $G_2=(\{2,3\},A_2\times A_3,(u_2,u_3))$, construct $G_T=(\{1,2,3\},A_1\times A_2\times A_3,(u_1,u_2,u_3))$ with $T=\{1,2\}$, $A_1=\{a_1\}$, and:
\begin{align*}
u_2(a_1,a_2,a_3)&=u_2(a_2,a_3),\quad
u_3(a_1,a_2,a_3)=u_3(a_2,a_3),\\
u_1(a_1,a_2,a_3)&=-(u_2(a_1,a_2,a_3)+u_3(a_1,a_2,a_3)).
\end{align*}

\textbf{($\Rightarrow$)} If $(x_T,x_3)$ is a CoE in $G_T$ with $x_T(a_1,a_2)=x_2(a_2)$:
\begin{itemize}
    \item Adversary best-response: $u_3(x_T,x_3)\geq u_3(x_T,x'_3)$ for all $x'_3\in X_3$ implies $u_3(x_2,x_3)\geq u_3(x_2,x'_3)$.
    \item CE constraint for player~2: Constraint  $x_T(a_1,a_2)(u_2(a_1,a_2,x_3)$ $-u_2(a_1,a'_2,x_3))\geq 0$ for all $a_2,a'_2\in A_2$ implies:
    \begin{equation*} 
\begin{split} 
       &x_2( a_2)(u_2(a_2,x_3)-u_2(a'_2,x_3))\geq 0,   \forall  a_2, a'_2  \\
     \Rightarrow & \sum_{a_2\in A_2} x_2( a_2)u_2(a_2,x_3)
  \geq \sum_{a_2\in A_2} x_2( a_2)u_2(a'_2,x_3),  \forall a'_2\in A_2  \\
    \Rightarrow & \sum_{a_2\in A_2} x_2( a_2)u_2(a_2,x_3)
  \geq   u_2(a'_2,x_3),  \forall a'_2\in A_2 \\
   \Rightarrow & \sum_{a'_2\in A_2} x'_2( a'_2)    \sum_{a_2\in A_2} x_2( a_2)u_2(a_2,x_3)\geq          \sum_{a'_2\in A_2} x'_2( a'_2)u_2(a'_2,x_3)   \\
    \Rightarrow &  \sum_{a_2\in A_2} x_2( a_2)u_2(a_2,x_3)\geq     \sum_{a_2\in A_2} x'_2( a_2)u_2(a_2,x_3), \forall x'_2\in X_2.
 \end{split}
 \end{equation*}
\end{itemize}
Thus $(x_2,x_3)$ is an NE in $G_2$.

\textbf{($\Leftarrow$)} If $(x_2,x_3)$ is an NE in $G_2$, then $(a_1,x_2,x_3)$ is an NE in $G_T$. By Theorem~\ref{theorem_nash_correlated_nash}, $(x_T,x_3)$ with $x_T(a_1,a_2)=x_2(a_2)$ is a CoE in $G_T$.
\end{proof}
\begin{Theorem}\label{thm:CoE-in-PPAD}
The problem of finding a CoE in an arbitrary finite adversarial team game belongs to PPAD.
\end{Theorem}


\begin{proof}
A CoE is a fixed point of the correspondence
\[
F:X_T\times X_n\;\rightrightarrows \;X_T\times X_n,
\qquad
F(x_T,x_n)
\;=\;
\bigl(\mathrm{BR}_T^{\mathrm{CE}}(x_n),\;\mathrm{BR}_n(x_T)\bigr),
\]
where
 \begin{itemize}
\item \(X_T=\Delta(A_T)\) and \(X_n=\Delta(A_n)\) are polytopes of polynomial description size,
\item \(\mathrm{BR}_n(x_T)\) is the set of best responses of the adversary to \(x_T\),
\item \(\mathrm{BR}_T^{\mathrm{CE}}(x_n)\) is the set of correlated strategies of the team that satisfy the CE constraints~\eqref{GCTMEConstraint} given \(x_n\).
\end{itemize}
Both best-response correspondences are non-empty, compact, convex-valued and upper hemi-continuous, so Kakutani’s theorem guarantees that \(F\) has a fixed point; every fixed point is a CoE.

The graph of \(F\) can be described by a polynomial number of linear and bilinear inequalities (the CE constraints become linear once \(x_n\) is fixed, and the adversary’s best-response inequalities become linear once \(x_T\) is fixed).
Consequently, the search for a fixed point of \(F\) reduces in polynomial time to the problem of finding an approximate Brouwer fixed point of a Lipschitz function given by an arithmetic circuit (or, equivalently, to the END-OF-LINE problem).  
Both of the latter problems are known to be in PPAD~\cite{papadimitriou1994complexity,daskalakis2009complexity,chen2009settling}, therefore so is the problem of finding a CoE.
\end{proof}

\begin{Corollary}\label{theorem_ppda}
The problem of finding a CoE is PPAD-complete, even in three-player zero-sum adversarial team games.
\end{Corollary}
\begin{proof}
Finding an NE in two-player games is PPAD-complete~\cite{chen2006settling,daskalakis2009complexity}. Theorem~\ref{reductionfromzpto30p} provides a polynomial-time reduction from two-player NE to CoE in three-player zero-sum adversarial team games, establishing PPAD-hardness. Membership in PPAD follows from Theorem \ref{thm:CoE-in-PPAD}.
\end{proof}

\subsection{Team-Maximizing CoE}
In adversarial team games, we are often interested in maximizing the team's welfare. We define the \emph{Team-Maximizing CoE} (TMCoE) as a CoE that maximizes the team's utility $u_T(x_T,x_n)=\sum_{a_n\in A_n}x_n(a_n)\sum_{\mathbf{a}_T\in\mathbf{A}_T}u_T(\mathbf{a}_T,a_n)x_T(\mathbf{a}_T)$ among all CoEs.

A TMCoE can be obtained by solving the following nonlinear program: $\max_{x_T}\min_{x_n}u_T(x_T,x_n)$, with variable $U^\star_T$ to represent the optimal value of $ u_T(x_T,x_n)$, i.e.,
\begin{subequations}\label{GCTMENLP} 
 \begin{align}
 &\max_{x_T ,x_n }U^\star_T\\
&\quad u_n(x_T,x_n)\geq u_n(x_T,a_n)\quad \forall a_n\in A_n \label{GCTME_minimize1}\\
    &\quad   
   \sum_{\mathbf{a}_T\in \mathbf{A}_T} x_T (\mathbf{a}_T)=1, x_T (\mathbf{a}_T)\geq 0 \\
    &\quad    
   \sum_{a_n\in A_n}x_n (a_n)=1,x_n (a_n)\geq 0 \label{GCTME_cd} \\
  &   \sum_{\mathbf{a}_{-(i,n)}\in \mathbf{A}_{-(i,n)}}\!\!\!\!\!\!\!\!\!\! x_T(a_i,\mathbf{a}_{-(i,n)})(u_i(a_i,\mathbf{a}_{-(i,n)},x_n)-u_i(a'_i,\mathbf{a}_{-(i,n)},x_n))\geq 0 \nonumber \\
  & \quad\quad\quad\quad\quad\quad\quad\quad\quad\quad\quad\quad\quad\quad\quad\quad \forall i\in T, a_i,a'_i\in A_i. \label{GCTME_ce}
\end{align}
\end{subequations}
Constraints~\eqref{GCTME_minimize1}--\eqref{GCTME_ce} define the space of CoEs (Eq.~\eqref{GCTME_minimize1} is equivalent to the adversary best-response condition).

\begin{Theorem}\label{theorem:computingTMCoE}
 Solving Program~\eqref{GCTMENLP} yields a CoE that maximizes the team's utility.
\end{Theorem}
\begin{proof} Variable $U^\star_T$   represents the optimal value of $ u_T(x_T,x_n)$. Then,
Eq.(\ref{GCTME_minimize1})   makes sure that $x_n$ is a best response to $x_T$, i.e.,
\begin{equation*} 
\begin{split}
     & u_n(x_T,x_n)\geq u_n(x_T,a_n), \forall a_n\in A_n \\
    \Rightarrow&  u_n(x_T,x_n)\geq \sum_{a_n}x'_n(a_n)u_n(x_T,a_n), \forall x'_n\in X_n\\
     \Rightarrow& u_n(x_T,x_n)\geq  u_n(x_T,x'_n), \forall x'_n\in X_n.
\end{split}
\end{equation*}
That is, $x_n$ is a best response to $x_T$. Given $x_n$, $x_T$ is a correlated equilibrium according to Eq.(\ref{GCTME_ce}). Then Eqs.(\ref{GCTME_minimize1})--(\ref{GCTME_ce}) define the space of CoEs. Therefore, solving Program (\ref{GCTMENLP}) obtains a CoE that maximizes the team's utility.
\end{proof}

\begin{Corollary}\label{theorem_nphard}
 Computing a TMCoE is NP-hard, even in zero-sum three-player adversarial team games.
\end{Corollary}
\begin{proof}
It has been shown that computing an NE maximizing the sum of players' utilities in two-player games is NP-hard \cite{conitzer2003complexity}.   To show that computing a TMCoE is NP-hard, we reduce each two-player game into a three-player  adversarial team game. Then, exploiting the complexity result in \cite{conitzer2003complexity}, we can     show that computing a TMCoE is NP-hard in three-player adversarial team games.

Formally, for any two-player game $G_2=(\{2,3\},A_2\times A_3,(u_2,u_3))$, we construct a three-player zero-sum adversarial team game $G_T=(\{1,2,3\},A_1\times A_2\times A_3,(u_1,u_2,u_3))$ such that the team is $T=\{1,2\}$, the adversary is player $3$, the action space and utility function for player 2 or player 3 are the same in both $G_2$ and $G_T$,  respectively, $A_1=\{a_1\}$,  and    $u_1=-(u_2+u_3)$. That is, we add a dummy team player with only one action and the utility is the negative value of the sum of other players' utilities:  $\forall a_2\in A_2,a_3\in A_3, u_2(a_1,a_2,a_3)=u_2(a_2,a_3), u_3(a_1,a_2,a_3)=u_3(a_2,a_3),   u_1(a_1,a_2,a_3)=-(u_2(a_1,a_2,a_3)+u_3(a_1,a_2,a_3)).$ 

We can verify that, $(x_2,x_3)$ is an NE in $G_2$ if and only if $(x_T,x_3)$ with $x_T(a_1,a_2)=x_2(a_2)$ for each $a_2\in A_2$ is a CoE in $G_T$ (see  Theorem \ref{reductionfromzpto30p}).


Then, we can   use the complexity result in \cite{conitzer2003complexity}  to show that computing a TMCoE is NP-hard in three-player adversarial team games. However, to show that computing a TMCoE is NP-hard in `zero-sum' three-player adversarial team games, we need to revise the proof in \cite{conitzer2003complexity}.  In the above reduction from a two-player game to a three-player adversarial team game, to obtain a zero-sum three-player adversarial team game, we have to make sure that the  dummy team player's utility    is the negative value of the sum of other players' utilities. 
 Then, maximizing the sum of utilities of players $2$ and $3$ may not maximize the team's utilities. In fact, to maximize the team's utilities, we need to minimize the utility of player $3$ in $G_T$. 
  If $G_2$ is a symmetric game, we still need to minimize the sum of utilities of players $2$ and $3$ in $G_2$.  In  \cite{conitzer2003complexity}, the reduction constructs a two-player symmetric game from each Boolean formula in the conjunctive normal form $\phi$, which shows that there is  a  Nash  equilibrium  in that game with utility  1  for each player if  and  only  if $\phi$ is satisfiable; otherwise, the only equilibrium makes each player get 0. Obviously,  this utility setting is not suitable for showing  that computing an NE minimizing the player's utility is NP-hard. Therefore, we need to revise the proof in \cite{conitzer2003complexity} to show that computing an NE minimizing the player's utility is NP-hard.

We construct a two-player symmetric game from each Boolean formula in the conjunctive normal form $\phi$ (e.g., $(z_1\wedge z_3\wedge \neg z_5)\cup(z_1\wedge \neg z_3\wedge   z_5)$). $V=\{y_1,\dots,y_k\}$ is a set of variables in $\phi$ with $|V|=k$,  $L=\{z_1,\neg z_1,\dots,z_k,\neg z_k\}$ is the set of literals (each variable corresponds to a positive literal and a negative literal), and $C$ is the set of clauses (e.g., $C=\{(z_1\wedge z_3\wedge \neg z_5), (z_1\wedge \neg z_3\wedge   z_5)\}$). $v:L\rightarrow V$ is a function mapping each literal to its corresponding variable, i.e., $v(z_1)=v(\neg z_1)=y_1$. $G(\phi)$ is a symmetric two-player normal-form game such that $A'=A_1=A_2=L\cup V\cup C\cup \{f\}$ with the following utility functions: \begin{align*}
    & u_1(l_1,l_2)=u_2(l_2,l_1)=1 \quad \forall l_1,l_2 \in L, s.t. l_1\neq l_2;\\
    &u_1(l,\neg l)=u_2(\neg l,l) = -2 \quad \forall l\in L;\\
    &u_1(l,x)=u_2(x,l)=-2 \quad \forall l\in L, x\in A'\setminus L;\\
    &u_1(y,l)=u_2(l,y) =2 \quad \forall y\in V,l\in L s.t. v(l)\neq y;\\
     &u_1(y,l)=u_2(l,y) =2-k \quad \forall y\in V,l\in L, s.t. v(l)= y;\\
      &u_1(y,x)=u_2(x,y) =-2 \quad \forall y\in V,x\in A'\setminus L;\\
        &u_1(c,l)=u_2(l,c) =2 \quad \forall c\in C,l\in L, s.t. l\notin c;\\
   &u_1(c,l)=u_2(l,c) =2-k \quad \forall c\in C,l\in L, s.t. l\in c;\\
      &u_1(c,x)=u_2(x,c) =-2 \quad \forall c\in C,x\in A'\setminus L;\\
     & u_1(f,f)=u_2(f,f)=2;\\
    &  u_1(f,x)=u_2(x,f)=1  \quad \forall  x\in A'\setminus \{f\}.
\end{align*}
Similar to the proof in \cite{conitzer2003complexity}, we show that if $(l_1,l_2,\dots,l_k)$ with $v(l_i)=y_i$ satisfies $\phi$, then there is an NE in $G(\phi)$ where each player plays $l_i$ with probability $\frac{1}{k}$ with expected utility 1, and vice versa. The only other Nash equilibrium is  the strategy profile that each player plays $f$ with utility $2$ for each player. 
Obviously, this strategy profile that each player plays $f$ is always an NE.

Now, we verify that the strategy profile where each player plays each $l_i$ with probability $\frac{1}{k}$ is an NE.
If $(l_1,l_2,\dots,l_k)$ with $v(l_i)=y_i$ satisfies $\phi$, we show that the strategy profile where each player plays $l_i$ with probability $\frac{1}{k}$ is an NE.
We need to show: if one player plays this strategy,  playing this strategy (i.e., playing each $l_i$ with probability $\frac{1}{k}$)  is a best response of the other player. We can check it one by one: 1) playing any $l_i$: the utility is 1; 2) playing any negation of $l_i$: the utility is $\frac{-2}{k}+\frac{k-1}{k}<1$; 3) playing any $y\in V$: the utility is $\frac{2-k}{k}+\frac{2k-2}{k}=1$ because there is $l_i$ such that $v(l_i)=y$; 4) playing any $c\in C$:  the utility is at most $\frac{2-k}{k}+\frac{2k-2}{k}=1$ because there is at least one literal $l_i\in c$; and 5) playing $f$: the utility is 1. Then, if one player plays this strategy,  playing this strategy (i.e., playing each $l_i$ with probability $\frac{1}{k}$) is a best response for the other player.

Now, we show that there are no more NEs.
If one player plays $f$, then the only best response of the other player is playing $f$. Then there are no NEs such that only one player always plays $f$. Suppose   players both play $f$ with the probability that is less than one. Let us consider  the expected social welfare (the sum of each player's expected utility) with that   players do not play $f$. It is clear that there is no outcome with this social welfare greater than 2 (e.g., if $u_1(y,l)=2$, $u_2(y,l)=-2$). Moreover, any outcome involving  an element of $V$ or $C$ has  social welfare strictly less than 2. Then, if any player plays an element of $V$ and $C$, the expected social welfare is strictly below 2, which means that at least one player's expected utility is strictly less than 1, given that   players do not play $f$. Then, this player will deviate from it to play $f$. Therefore, in any Nash equilibrium, no element in $V$ or $C$ is played.

Now, we can assume players only play elements in $L\cup \{f\}$ with positive probability. If one player plays $f$ with positive probability, then playing $f$ is the strict best response for the other player. Then, if $f$ is played in an NE, $(f,f)$ is that Nash equilibrium.

Now, we assume there is an NE where both players only play elements in $L$ with positive probability. If there is $l\in L$ such that $l$ or $\neg l$  is played with probability less than $\frac{1}{k}$, then playing $v(l)$ for the other player obtains the utility more than $\frac{2-k}{k}+\frac{2k-2}{k}=1$, and then this strategy profile cannot be an NE. Then if $l\in L$ or $\neg l$ is played, the probability is $\frac{1}{k}$.

If there is $l\in L$ that player 1 plays it with positive probability but player 2 plays its negation, each player will have the expected utility less than 1, which is dominated by playing $f$. Then if $l_i$ is played by any player, both players will play it with the probability $\frac{1}{k}$. We then can assume that exactly one of each variable's corresponding literals is played with the probability $\frac{1}{k}$ by both players.
Then, in any Nash equilibrium, if literals are played with positive probability, they correspond to an assignment to the variables.

If $\phi$ is not satisfied, then there is a clause $c\in C$, which is not satisfied, i.e., none of the literals in $c$ is played. Then, playing $c$ will be better than playing any $l_i$ for both players. Then the strategy profile playing $l_i$ with $\frac{1}{k}$ is not an NE.

Therefore, there exists an NE in $G(\phi)$ with the utility 1 for each player if and only if $\phi$ is satisfiable; otherwise, the only equilibrium gives the utility 2 for each player. Then, finding an NE minimizing a player's utility is NP-hard in two-player games, and then finding  a TMCoE is NP-hard in zero-sum three-player adversarial team games.
\end{proof}

 Remark. Although computing a CE (cooperation without competition) requires only linear programming, our results show that adding the adversary (competition) makes CoE as hard as NE: PPAD-complete to find, NP-hard to optimize. This reveals a fundamental gap between pure cooperation and co-opetition, and motivates the search for tractable subclasses.

\section{Efficient Computation in Consistent Games}\label{section_consistent}
 
The NP-hardness result motivates identifying games where TMCoEs can be computed efficiently. We characterize a broad class of zero-sum adversarial team games where this is possible:  TMCoEs can be computed in polynomial time via linear programming with the property of exchangeability.

\subsection{Consistent Constraints}
 
In two-player zero-sum games, NEs are exchangeable and computable via linear programming. We show that zero-sum adversarial team games retain these properties when each team player's utility is \emph{consistent} with the team's utility---a condition that generalizes the standard identical-utility assumption.

\begin{Definition}\label{def:consistent}
 A zero-sum adversarial team game satisfies \textbf{consistent constraints} if there exists $k_i>0$ such that for all $i\in T$, $a_i,a'_i\in A_i$, $x_T\in X_T$, $x_n\in X_n$:
\begin{equation}\label{eqconsistentconstraint} 
\begin{split}
     \sum_{\mathbf{a}_{-(i,n)}\in \mathbf{A}_{-(i,n)}} \!\!\!\!\!\!\!\!\!\!x_T(a_i,\mathbf{a}_{-(i,n)})(u_T(a_i,\mathbf{a}_{-(i,n)},x_n)-u_T(a'_i,\mathbf{a}_{-(i,n)},x_n))   \\
  = k_i \sum_{\mathbf{a}_{-(i,n)}\in \mathbf{A}_{-(i,n)}} \!\!\!\!\!\!\!\!\!\!x_T(a_i,\mathbf{a}_{-(i,n)})(u_i(a_i,\mathbf{a}_{-(i,n)},x_n)-u_i(a'_i,\mathbf{a}_{-(i,n)},x_n)).
\end{split}
\end{equation}
\end{Definition}

Eq.~\eqref{eqconsistentconstraint} requires that each team player's incentive to deviate aligns proportionally (by factor $k_i$) with the team's incentive. This holds, for example, when $u_i(\mathbf{a}_T)=k_i u_T(\mathbf{a}_T)$ with $k_i>0$ for all outcomes---i.e., each player's utility is a positive fraction of the team's utility, but these fractions need not be equal across players. This substantially generalizes the identical-utility assumption. 

\begin{Theorem}\label{theorem_specailzerosumgame}
In zero-sum adversarial team games satisfying consistent constraints (Eq.~\ref{eqconsistentconstraint}), a TMCoE is computed by the following linear program and its dual program: 
\begin{subequations}\label{ne_linear} 
 \begin{align}
 &\max_{x_T}  U^\star_T\\
&\quad U^\star_T\leq u_T(x_T,a_n)\quad \forall a_n\in A_n \label{specialne_minimize1}\\
     &\quad    
  \sum_{\mathbf{a}_T\in \mathbf{A}_T} x_T (\mathbf{a}_T)=1, x_T (\mathbf{a}_T)\geq 0  \label{ne_1} 
\end{align}
\end{subequations}
\end{Theorem}
\begin{proof}
Obviously, the space of $x_T$ in Eqs.(\ref{specialne_minimize1}) and (\ref{ne_1}) includes the space      of $x_T$ in Eqs.(\ref{GCTME_minimize1})--(\ref{GCTME_ce}) with additional constraints shown in Eq.(\ref{GCTME_ce}), and Eq.(\ref{GCTME_minimize1}) is equivalent to Eq.(\ref{specialne_minimize1}) in zero-sum games. Program (\ref{ne_linear}) as a maxmin program computes an NE for two-player zero-sum games.
If $(x_T,x_n)$  is an optimal solution in Program (\ref{ne_linear}) and its dual program, i.e., $(x_T,x_n)$ is an NE between the team and the adversary,  then for each $i\in T$, $a_i,a'_i\in A_i$, $\mathbf{a}_{-(i,n)}\in \mathbf{A}_{-(i,n)}$, we have: If $x_T(a_i,\mathbf{a}_{-(i,n)}) >0$, 
\begin{equation*} 
\begin{split} 
& u_T(x_T,x_n)=  u_T(a_i,\mathbf{a}_{-(i,n)},x_n) \geq  u_T(a'_i,\mathbf{a}_{-(i,n)},x_n) \\
  \Rightarrow &  x_T(a_i,\mathbf{a}_{-(i,n)})u_T(a_i,\mathbf{a}_{-(i,n)},x_n)  
   \!\geq\! 
 x_T(a_i,\mathbf{a}_{-(i,n)}\!) u_T(a'_i,\mathbf{a}_{-(i,n)},x_n). 
\end{split}
\end{equation*}
The above relation holds as well when    $x_T(a_i,\mathbf{a}_{-(i,n)}) =0$. Therefore, $\forall i\in T, a_i,a'_i\in A_i$,
\begin{equation*} 
       \sum_{\mathbf{a}_{-(i,n)}\in \mathbf{A}_{-(i,n)}}\!\!\!\!\!\!\!\! x_T(a_i,\mathbf{a}_{-(i,n)})(u_T(a_i,\mathbf{a}_{-(i,n)},x_n)    -
 u_T(a'_i,\mathbf{a}_{-(i,n)},x_n))\geq 0.
\end{equation*}
By Eq.(\ref{eqconsistentconstraint}), we have: $\forall i\in T, a_i,a'_i\in A_i$,
\begin{equation*} 
       \sum_{\mathbf{a}_{-(i,n)}\in \mathbf{A}_{-(i,n)}}\!\!\!\!\!\!\!x_T(a_i,\mathbf{a}_{-(i,n)})(u_i(a_i,\mathbf{a}_{-(i,n)},x_n)   -
 u_i(a'_i,\mathbf{a}_{-(i,n)},x_n))\geq 0.
\end{equation*}
$x_n$ is a best response to $x_T$ as well. Therefore, $(x_T,x_n)$ is in the space of Eqs.(\ref{GCTME_minimize1})--(\ref{GCTME_ce}). Then $(x_T,x_n)$  is an optimal solution of Program (\ref{GCTMENLP}), i.e., a TMCoE. That is, each optimal solution of Program (\ref{ne_linear}) is an optimal solution of Program (\ref{GCTMENLP}), and both programs have the same optimal objective value. 
\end{proof}

\begin{Theorem}\label{therorem_equivalent}
 In zero-sum adversarial team games satisfying consistent constraints (Eq.~\ref{eqconsistentconstraint}), a TMCoE $(x_T,x_n)$ is equivalent to a two-player zero-sum NE (team as a single player vs.\ adversary), and vice versa.
\end{Theorem}
\begin{proof}
 The feasible region of~\eqref{ne_linear} contains that of~\eqref{GCTMENLP} (since~\eqref{GCTME_ce} are additional constraints). By Theorem~\ref{theorem_specailzerosumgame}, every optimal solution of~\eqref{ne_linear} is optimal for~\eqref{GCTMENLP}, so the CE constraints are redundant. Thus the TMCoE coincides with the two-player NE.  
\end{proof}

\subsection{Exchangeability and Polynomial-Time Computation}
When consistent constraints hold, TMCoEs inherit the exchangeability property of two-player zero-sum NEs, eliminating equilibrium selection problems.

\begin{Theorem}\label{thm:exchangeable}
 In zero-sum adversarial team games satisfying consistent constraints, TMCoEs are exchangeable: if $(x_T,x_n)$ and $(x'_T,x'_n)$ are both TMCoEs, then $(x'_T,x_n)$ and $(x_T,x'_n)$ are also TMCoEs.
\end{Theorem}
\begin{proof}
 By Theorem~\ref{therorem_equivalent}, $(x_T,x_n)$ and $(x'_T,x'_n)$ are two-player zero-sum NEs. The standard exchangeability argument gives $u_T(x_T,x_n)\leq u_T(x_T,x'_n)\leq u_T(x'_T,x'_n)\leq u_T(x'_T,x_n)\leq u_T(x_T,x_n)$, hence all four have equal team utility and are NEs, and thus TMCoEs.
\end{proof}

 When consistent constraints fail, TMCoEs may not be exchangeable. Table~\ref{tab:team_adv_game2} provides a counterexample.

\begin{table}[h]
\centering 
\begin{tabular}{|c|c|c|c|c|c|c|c|}
        \cline{1-3}\cline{5-7} $c_1$&$b_1$&$b_2$&&$c_2$&$b_1$&$b_2$  \\\cline{1-3}\cline{5-7}
         $a_1$&0,0,0&0,1,-1&&$a_1$&0,0,0&0,1,-1\\\cline{1-3}\cline{5-7}
         $a_2$&1,0,-1&0,0,0&&$a_2$&1,0,-1&0,1,-1\\\cline{1-3}\cline{5-7}
    \end{tabular}
       \caption{ A zero-sum adversarial team game where CoEs are not exchangeable: $(a_2,b_1,c_1)$ and $(a_1,b_2,c_2)$ are CoEs, but $(a_2,b_1,c_2)$ is not (since $b_1$ is dominated by $b_2$ when the adversary plays $c_2$). Consistent constraints (Eq.~\ref{eqconsistentconstraint}) fail here: for $(a_2,b_1,c_1)$ vs.\ $(a_2,b_2,c_1)$, team utilities are 1 and 0 but player~2's utilities are both 0, so no $k>0$ satisfies $(1-0)=k(0-0)$.}
    \label{tab:team_adv_game2}
\end{table}

\begin{Theorem}\label{theorem_polynomial}
 In zero-sum adversarial team games satisfying consistent constraints, a TMCoE can be computed in polynomial time.
\end{Theorem}
\begin{proof}
 By Theorem~\ref{therorem_equivalent}, computing a TMCoE reduces to computing a two-player zero-sum NE. Program~\eqref{ne_linear} can be reformulated as $\max_{x_T\in X_T}\min_{x_n\in X_n}u_T(x_T,x_n)$, whose dual over $x_n$ yields a linear program solvable in polynomial time~\cite{shoham2008multiagent}.
\end{proof}

\textbf{Key insight.} In these consistent games, team players' utilities need not be identical---they only need to align directionally with the team's utility. This substantially broadens the scope of tractable adversarial team games beyond the classical identical-utility setting. Notably, if team players act independently (NE), computing an optimal equilibrium (TME) remains NP-hard even under consistent constraints: TME is NP-hard already for identical-utility normal-form games~\cite{basilico2017team}, and consistent constraints strictly generalize that class. Thus, it is the ability to correlate team actions---and thereby treat the team as a single player optimizing $u_T$---that is essential for tractability, not the consistent-constraint condition alone.

\section{Discussion}\label{sec_discussion}

\subsection{Cooperation and Competition}
The adversarial team game structure naturally encodes cooperation (among team members) and competition (against the adversary). Prior TME and CTME concepts capture this only under identical team utilities~\cite{von1997team,basilico2017team,celli2018computational,Anagnostides2023Algorithms}. In the economics literature, ``co-opetition'' describes business strategies with simultaneous cooperation and competition among competitors~\cite{brandenburger1996co,bengtsson2000coopetition,bouncken2015coopetition,okura2014coopetition,sun2005coopetitive}---this is conceptually distinct from our CoE, which describes the \emph{structure} of players: cooperation within the team, competition between team and adversary.

Coalitional games with transferable/nontransferable utility~\cite{shoham2008multiagent} model cooperative scenarios but abstract away individual actions. Our model retains individual action spaces and incentive constraints. When consistent constraints hold and the adversary's policy is fixed, the team subgame becomes a potential game~\cite{shoham2008multiagent}; however, the full game is not a potential game (nor an adversarial potential game~\cite{Anagnostides2023Algorithms}) because the adversary's policy is endogenous.

\subsection{Relationship to Concave Games}
Our CoE differs from a Rosen-NE~\cite{rosen1965existence} in the two-player (team, adversary) metagame. A Rosen-NE would require the team to solve $\max_{x_T}u_T(x_T,x_n)$ subject to CE constraints, whereas CoE only requires the CE constraints to hold---no team-side optimization. This additional maximization operator in Rosen-NE suggests that computing a Rosen-NE should be \emph{harder} than computing a CoE. Yet, surprisingly, both are PPAD-complete.

\subsection{Complexity Remarks}
Our complexity results build on reductions from two-player NE~\cite{chen2006settling,daskalakis2009complexity,gilboa1989nash,conitzer2003complexity}, but are not straightforward applications: we must reconstruct a CoE from an NE (and vice versa) in an adversarial team game where the adversary is not a dummy player---a construction that is impossible under identical team utilities. Our CoE is fundamentally different from NE: (i)~team players correlate actions (CoE) rather than playing independently (NE); (ii)~in consistent games (Section~\ref{section_consistent}), TMCoEs are exchangeable and efficiently computable, unlike NEs. Moreover, CTME is computable via linear programming~\cite{basilico2017team} but ignores utility differences; approximate NE computation in adversarial team games is CLS-complete under identical-utility assumptions~\cite{Anagnostides2023Algorithms}, conditions that do not apply here.

\subsection{Real-World Applications}
Our framework models many real-world scenarios: (1)~cooperation-competition markets, where firms with heterogeneous objectives ally against new entrants; (2)~security games with multiple defenders at ports or urban networks~\cite{jiang2013defender,song2023multi,Jain11,Jain2013,zhang2017optimal,zhang2019optimal,li2021cfr,xue2021solving,li2024grasper,zhuang2025solving}; (3)~interdependent defense games~\cite{chan2012interdependent}, where individual firms invest in protection against terrorists or hackers. Existing models either use NE among defenders~\cite{chan2012interdependent,song2023multi} or assume identical defender utilities~\cite{jiang2013defender,Jain11,Jain2013,zhang2017optimal,zhang2019optimal,li2023solving,li2023decision}. CoE provides a more faithful solution concept for  heterogeneous-defender settings.

\subsection{Extension to Multi-Team Games}
Our formulation assumes a single adversary, consistent with much of the adversarial team game literature~\cite{von1997team,basilico2017team,celli2018computational,zhang2020computing,farina2021connecting,zhang2022correlation,zhang2022subgame} and many real-world scenarios. When the adversary itself is a team satisfying consistent constraints, it can be treated as a single player and our results carry over directly. For general multi-team games where each team plays a CE internally while best-responding to other teams, CoEs exist and our results extend when consistent constraints hold within each team. When they do not, existence is not guaranteed---the team's best response may conflict with individual members' best responses---and identifying the right solution concept for this setting is an important open direction.

\subsection{Relationship to Heterogeneous MARL}

Our work is closely connected to Multi-Agent Reinforcement Learning (MARL) \cite{shi2019learning,Wang2026shapley} with heterogeneous agents. We briefly survey existing paradigms and explain why they do not---and in some cases cannot---compute a CoE.

\textbf{Cooperative MARL with parameter sharing.} A dominant paradigm in cooperative MARL is parameter sharing~\cite{lowe2017multi,yu2022surprising}, where all agents share a single policy network, implicitly assuming homogeneous agents with identical reward functions. While computationally effective, these methods cannot represent heterogeneous team utilities: the shared policy optimizes the team's aggregate return but ignores individual incentive constraints. In CoE terms, this corresponds to solving the two-player zero-sum LP (Program~\ref{ne_linear}) without the CE constraints~\eqref{GCTME_ce}---which is optimal only when consistent constraints hold.

\textbf{Heterogeneous-agent MARL with shared team rewards.} Two recent works study ``heterogeneous'' team games but in a sense orthogonal to ours. Zhong et al.~\cite{zhong2024heterogeneous} develop MARL algorithms (HATRPO, HAPPO) for agents with heterogeneous action spaces, observation functions, and policy parameters, but under a single joint reward function shared by all team members. Liu et al.~\cite{liu2024computing} study ex ante equilibrium computation in heterogeneous zero-sum team games, but define each team member's reward as the team reward divided equally ($R_{k,i} = R_k / n_k$). In both frameworks, ``heterogeneity'' refers to \emph{different capabilities} (action/observation spaces, policy architectures), not \emph{different objectives}: all team members share the same underlying reward function up to scaling. This is fundamentally different from our setting, where team players have genuinely distinct utility functions---as in the chicken-game example where one player receives $7$ and the other $2$ under the same outcome. In their frameworks, such utility differences are averaged away, collapsing to the identical-utility case that we show can destabilize equilibria (Section~\ref{section:cost}). Their methods consequently target NE or cooperative policies, not CoE with individual CE constraints. Value factorization methods like QMIX~\cite{rashid2018qmix} share the same limitation, decomposing a common team reward under the IGM principle.

\textbf{General-sum MARL and Nash-Q.} In general-sum stochastic games, Nash-Q learning~\cite{hu2003nash} and its variants aim to compute NE among players with distinct utilities. However, these methods require agents to play independent strategies and converge to NE, not CoE. The team in our setting correlates its actions, which is more expressive than independent play. No existing MARL algorithm explicitly enforces CE constraints for a subset of players (the team) while simultaneously best-responding to another player (the adversary).

\textbf{Team-PSRO and adversarial training.} McAleer et al.~\cite{mcaleer2023teampsro} propose Team-PSRO, which combines Policy Space Response Oracles (PSRO)~\cite{lanctot2017unified} with cooperative RL to approximate TMECor in large team games. However, Team-PSRO assumes identical team utilities, and the TMECor concept is a zero-sum NE between the team (as a single correlated player) and the adversary---it does not enforce individual CE constraints for team members with heterogeneous utilities. Adapting the PSRO framework to compute CoE would require the team's best-response oracle to solve a constrained optimization problem (maximizing team utility subject to CE constraints given the adversary's current strategy), which is precisely the nonlinear program~\eqref{GCTMENLP} or its consistent-constraint LP relaxation.

Therefore, to the best of our knowledge, no existing MARL algorithm computes a CoE in general adversarial team games with heterogeneous utilities. Existing work on ``heterogeneous'' MARL addresses heterogeneity in agent capabilities (action/observation spaces), not in objectives (utility functions). The key barriers are: (i)~most cooperative MARL assumes a shared team reward, (ii)~even work labeled ``heterogeneous'' preserves equal per-member rewards up to scaling, (iii)~general-sum MARL targets NE rather than correlated team play, and (iv)~adversarial training frameworks lack individual CE constraints for team members with divergent utilities.

\subsection{Future Work}
Several research directions emerge from our framework:
(i) \textbf{Efficient TMCoE algorithms.} For games requiring the nonlinear program~\eqref{GCTMENLP}, nonlinear optimization techniques \cite{sandholm2005mixedx,zhang2023computing,zhang2023finding} could be used to compute accurate equilibria, and stochastic optimization techniques from machine learning~\cite{jin2026adam,zhuang2025tree,mengreducing2025,gempapproximating2024,zeng2024computing,yu2024improving} or LLM techniques \cite{liu2026solving,kempinski2025game,zhao2025siren,yu2024direct,chen2024vision}  may yield practical approximations. For consistent games, MARL frameworks like PSRO~\cite{lanctot2017unified} and regret matching/counterfactual regret minimization   algorithms \cite{zinkevich2008regret,meng2026fastericlr,meng2026faster,meng2025last,meng2025efficient}  could scale to large real-world instances.
   (ii) \textbf{New solution concepts.} For multi-team games without consistent constraints, and for settings where heterogeneous team players act independently or consider fairness \cite{li2026online}, new equilibrium concepts are needed.
    (iii) \textbf{Broader tractable classes.} Identifying relaxations of consistent constraints that still admit efficient TMCoE computation is an open theoretical question.
     (iv) \textbf{Extensive-form games.} Extending CoE to extensive-form settings, exploiting sequential structure for compact representations, is a natural and important direction.

\section{Conclusion}
Motivated by real-world scenarios where collaborating parties have heterogeneous objectives---exemplified by UN sustainable development negotiations---we argue that adversarial team games must account for differences in team players' utility functions. Ignoring these differences can destabilize computed equilibria and cause significant losses. We introduced the Co-opetition Equilibrium (CoE) and its team-maximizing variant (TMCoE), which overcome this limitation. We proved that CoEs always exist (via reduction from NE), that finding a CoE is PPAD-complete, and that computing a TMCoE is NP-hard. Crucially, we identified the class of zero-sum adversarial team games with consistent constraints---generalizing identical utilities---where TMCoEs are exchangeable and computable in polynomial time via linear programming. We believe that accounting for heterogeneous utilities among team players opens a rich research agenda spanning equilibrium computation, algorithm design, and MARL.

\bibliographystyle{ACM-Reference-Format}
\bibliography{TmeEfg}

@inproceedings{meng2026faster,
  title={Faster game solving via asymmetry of step sizes},
  author={Meng, Linjian and Yang, Tianpei and Zhang, Youzhi and Ge, Zhenxing and Gao, Yang},
  booktitle={AAAI},
  volume={40},
  number={20},
  pages={17161--17169},
  year={2026}
}

@inproceedings{li2026online,
  title={Online fair allocation with demand-side time-dependent weight},
  author={Li, Minming and Zhang, Youzhi and Zheng, Shangkun},
  booktitle={UAI},
  year={2026}
}

@inproceedings{meng2025last,
  title={Last-Iterate Convergence of Smooth Regret Matching$^{+}$ Variants in Learning {N}ash Equilibria},
  author={Meng, Linjian and Zhang, Youzhi and Ge, Zhenxing and Ding, Tianyu and Yang, Shangdong and Xu, Zheng and Li, Wenbin and Gao, Yang},
  booktitle={NeurIPS},
  volume={38},
  pages={19194--19239},
  year={2025}
}

@inproceedings{meng2025efficient,
  title={Efficient Last-Iterate Convergence in Solving Extensive-Form Games},
  author={Meng, Linjian and Yang, Tianpei and Zhang, Youzhi and Ge, Zhenxing and Yang, Shangdong and Ding, Tianyu and Li, Wenbin and An, Bo and Gao, Yang},
  booktitle={NeurIPS},
  volume={38},
  pages={100628--100675},
  year={2025}
}

@inproceedings{liu2026solving,
  title={Solving Repeated Games with Large Language Model},
  author={Liu, Naming and Zhang, Youzhi and Wen, Ying},
 booktitle={AAMAS},
  pages={569--577},
  year={2026}
}

@article{yu2024direct,
  title={Direct alignment of language models via quality-aware self-refinement},
  author={Yu, Runsheng and Wang, Yong and Jiao, Xiaoqi and Zhang, Youzhi and Kwok, James T},
  journal={arXiv preprint arXiv:2405.21040},
  year={2024}
}

@ARTICLE{Wang2026shapley,
  author={Wang, Wanyuan and Che, Qian and Liu, Chunyu and Zhang, Youzhi and Jiang, Jiuchuan and An, Bo},
  journal={IEEE Transactions on Automation Science and Engineering}, 
  title={Shapley Meets {DCOP}: {A} Unified Structural Credit Assignment for Multiagent Planning and Multiagent Reinforcement Learning}, 
  year={2026},
  volume={23},
  pages={3967-3983}
  }

@inproceedings{shi2019learning,
  title={Learning expensive coordination: {A}n event-based deep {RL} approach},
  author={Shi, Zhenyu and Yu, Runsheng and Wang, Xinrun and Wang, Rundong and Zhang, Youzhi and Lai, Hanjiang and An, Bo},
  booktitle={ICLR},
  year={2019}
}

@article{li2023decision,
  title={Decision making in team-adversary games with combinatorial action space},
  author={Li, Shuxin and Zhang, Youzhi and Wang, Xinrun and Xue, Wanqi and An, Bo},
  journal={CAAI Artificial Intelligence Research},
  volume={2},
  pages={9150023},
  year={2023}
}

@inproceedings{yu2024improving,
  title={Improving sharpness-aware minimization by lookahead},
  author={Yu, Runsheng and Zhang, Youzhi and Kwok, James},
  booktitle={ICML},
  year={2024}
}

@inproceedings{chen2024vision,
  title={Vision language model empowered surgical planning},
  author={Chen, Yihe and Yu, Runsheng and Wang, Xin and Wang, Wensheng and Tan, Ning and Zhang, Youzhi},
  booktitle={2024 International Conference on Intelligent Robotics and Automatic Control (IRAC)},
  pages={317--322},
  year={2024},
  organization={IEEE}
}

@inproceedings{zhao2025siren,
  title={Siren: {A} learning-based multi-turn attack framework for simulating real-world human jailbreak behaviors},
  author={Zhao, Yi and Zhang, Youzhi},
  booktitle={2025 IEEE Annual Computer Security Applications Conference (ACSAC)},
  pages={1209--1220},
  year={2025},
  organization={IEEE}
}

@inproceedings{zhang2023finding,
  title={Finding optimal {N}ash equilibria in multiplayer games via correlation plans},
  author={Zhang, Youzhi and An, Bo and Subrahmanian, VS},
  booktitle={AAMAS},
  pages={2712--2714},
  year={2023}
}

@inproceedings{jin2026adam,
  title={Adam under Generalized Smoothness with Second-Moment-Type Stochastic Gradients},
  author={Jin, Ruinan and Cheng, Difei and Chen, Ling and Luo, Jun and Zhou, Hao and Zhang, Youzhi},
  booktitle={NeurIPS},
  year={2026}
}

@inproceedings{kempinski2025game,
  title={Game of thoughts: {I}terative reasoning in game-theoretic domains with large language models},
  author={Kempinski, Benjamin and Gemp, Ian and Larson, Kate and Lanctot, Marc and Bachrach, Yoram and Kachman, Tal},
  booktitle={AAMAS},
  pages={1088--1097},
  year={2025}
}

@inproceedings{meng2026fastericlr,
  title={A faster parameter-free regret matching algorithm},
  author={Meng, Linjian and Zhang, Youzhi and Yang, Shangdong and Li, Wenbin and Ding, Tianyu and Gao, Yang},
  booktitle={ICLR},
  volume={2026},
  pages={49326--49351},
  year={2026}
}

@inproceedings{zhang2023computing,
  title={Computing optimal {N}ash equilibria in multiplayer games},
  author={Zhang, Youzhi and An, Bo and Subrahmanian, Venkatramanan},
  booktitle={NeurIPS},
  volume={36},
  pages={21230--21254},
  year={2023}
}

@article{papadimitriou1994complexity,
  title={On the complexity of the parity argument and other inefficient proofs of existence},
  author={Papadimitriou, Christos H},
  journal={Journal of Computer and system Sciences},
  volume={48},
  number={3},
  pages={498--532},
  year={1994},
  publisher={Elsevier}
}

@article{chen2009settling,
  title={Settling the complexity of computing two-player {N}ash equilibria},
  author={Chen, Xi and Deng, Xiaotie and Teng, Shang-Hua},
  journal={Journal of the ACM (JACM)},
  volume={56},
  number={3},
  pages={1--57},
  year={2009},
  publisher={ACM New York, NY, USA}
}

@article{zhuang2025tree,
  title={Tree-Based Stochastic Optimization for Solving Large-Scale Urban Network Security Games},
  author={Zhuang, Shuxin and Meng, Linjian and Li, Shuxin and Li, Minming and Zhang, Youzhi},
  journal={AAAI},
  year={2026}
}

@inproceedings{zhang2024dag,
  title={{DAG}-Based Column Generation for Adversarial Team Games},
  author={Zhang, Youzhi and An, Bo and Zeng, Daniel Dajun},
  booktitle={ICML},
  pages={58563--58578},
  year={2024} 
}

@article{liu2024computing, 
title={Computing ex ante equilibrium in heterogeneous zero-sum team game},
  author={Liu, Naming and Wang, Mingzhi and Wang, Xihuai and Zhang, Weinan and Yang, Yaodong and Zhang, Youzhi and An, Bo and Wen, Ying},
  journal={Frontiers of Computer Science},
  volume={20},
  number={11},
  pages={2011364},
  year={2026}
}

@article{liu2024leveraging,
  title={Leveraging Team Correlation for Approximating Equilibrium in Two-Team Zero-Sum Games},
  author={Liu, Naming and Wang, Mingzhi and Zhang, Youzhi and Yang, Yaodong and An, Bo and Wen, Ying},
  journal={arXiv preprint arXiv:2403.00255},
  year={2024}
}

@inproceedings{zeng2024computing, 
title={Computing Approximate {N}ash Equilibrium in Two-Team Zero-Sum Games by {NashConv} Descent}, author={Zeng, Zekeng and Zhang, Youzhi and Yang, Peipei and Zhang, Mingyi and Zhang, Junge}, 
booktitle={ICONIP}, 
pages={167--181}, 
year={2024}  
}

@inproceedings{li2024grasper,
  title={Grasper: {A} Generalist Pursuer for Pursuit-Evasion Problems},
  author={Li, Pengdeng and Li, Shuxin and Wang, Xinrun and Cern{\`y}, Jakub and Zhang, Youzhi and McAleer, Stephen and Chan, Hau and An, Bo},
  booktitle={AAMAS},
  pages={1147--1155},
  year={2024}
}

@inproceedings{mengreducing2025, 
title={Reducing Variance of Stochastic Optimization for Approximating {N}ash Equilibria in Normal-Form Games}, 
author={Meng, Linjian and Chen, Wubing and Li, Wenbin and Yang, Tianpei and Zhang, Youzhi and Gao, Yang}, booktitle={ICML},
year={2025}  
}

@article{zhuang2025solving,
  title={{GraphChase}: {A} platform and benchmark for urban network security games},
  author={Zhuang, Shuxin and Li, Shuxin and Yang, Tianji and Li, Muheng and Shi, Xianjie and An, Bo and Zhang, Youzhi},
  journal={KDD},
  year={2026}
}

@inproceedings{gempapproximating2024,
  title={Approximating {N}ash Equilibria in Normal-Form Games via Stochastic Optimization},
  author={Gemp, Ian and Marris, Luke and Piliouras, Georgios},
  booktitle={ICLR},
year={2024}
}

@book{shoham2008multiagent,
  title={Multiagent Systems: {A}lgorithmic, Game-Theoretic, and Logical Foundations},
  author={Shoham, Yoav and Leyton-Brown, Kevin},
  year={2008},
  publisher={Cambridge University Press}
}

@article{von1997team,
  title={Team-maxmin equilibria},
  author={von Stengel, Bernhard and Koller, Daphne},
  journal={Games and Economic Behavior},
  volume={21},
  number={1-2},
  pages={309--321},
  year={1997},
  publisher={Elsevier}
}

@inproceedings{celli2018computational,
  title={Computational results for extensive-form adversarial team games},
  author={Celli, Andrea and Gatti, Nicola},
  pages={965--972},
  booktitle={AAAI},
  year={2018}
}

@inproceedings{basilico2017team,
  title={Team-maxmin equilibrium: {E}fficiency bounds and algorithms},
  author={Basilico, Nicola and Celli, Andrea and De Nittis, Giuseppe and Gatti, Nicola},
  pages={356--362},
  booktitle={AAAI},
  year={2017}
}

@inproceedings{farina2018ex,
  title={Ex ante coordination and collusion in zero-sum multi-player extensive-form games},
  author={Farina, Gabriele and Celli, Andrea and Gatti, Nicola and Sandholm, Tuomas},
  booktitle={NeurIPS},
  pages={9638--9648},
  year={2018}
}

@article{nash1951non,
  title={Non-cooperative games},
  author={Nash, John},
  journal={Annals of Mathematics},
  pages={286--295},
  year={1951},
  publisher={JSTOR}
}

@inproceedings{jiang2013defender,
  title={Defender (mis)coordination in security games},
  author={Jiang, Albert Xin and Procaccia, Ariel D and Qian, Yundi and Shah, Nisarg and Tambe, Milind},
  booktitle={IJCAI},
  pages={220--226},
  year={2013}
}

@inproceedings{Jain11,
 author = {Jain, Manish and Korzhyk, Dmytro and Van\v{e}k, Ond\v{r}ej and Conitzer, Vincent and P\v{e}chou\v{c}ek, Michal and Tambe, Milind},
 title = {A double oracle algorithm for zero-sum security games on graphs},
 booktitle = {AAMAS},
 year = {2011},
 pages = {327-334}
}

@inproceedings{Jain2013,
 author = {Jain, Manish and Conitzer, Vincent and Tambe, Milind},
 title = {Security Scheduling for Real-world Networks},
 booktitle = {AAMAS},
 year = {2013},
 isbn = {978-1-4503-1993-5},
 location = {St. Paul, MN, USA},
 pages = {215--222}
}

@inproceedings{zhang2020computing,
  title={Computing team-maxmin equilibria in zero-sum multiplayer extensive-form games},
  pages = {2318--2325},
  author={Zhang, Youzhi and An, Bo},
  booktitle={AAAI},
  year={2020}
}

@inproceedings{zinkevich2008regret,
  title={Regret minimization in games with incomplete information},
  author={Zinkevich, Martin and Johanson, Michael and Bowling, Michael and Piccione, Carmelo},
  booktitle={NeurIPS},
  pages={1729--1736},
  year={2008}
}

@inproceedings{zhang2019optimal,
  title={Optimal interdiction of urban criminals with the aid of real-time information},
  author={Zhang, Youzhi and Guo, Qingyu and An, Bo and Tran-Thanh, Long and Jennings, Nicholas R},
  booktitle={AAAI},
  volume={33},
  pages={1262--1269},
  year={2019}
}

@inproceedings{lanctot2017unified,
  title={A unified game-theoretic approach to multiagent reinforcement learning},
  author={Lanctot, Marc and Zambaldi, Vinicius and Gruslys, Audrunas and Lazaridou, Angeliki and Tuyls, Karl and P{\'e}rolat, Julien and Silver, David and Graepel, Thore},
  booktitle={NeurIPS},
  pages={4190--4203},
  year={2017}
}

@inproceedings{zhang2020converging,
  title={Converging to team-maxmin equilibria in zero-sum multiplayer games},
  author={Zhang, Youzhi and An, Bo},
  booktitle={ICML},
  pages={11033--11043},
  year={2020}
}

@inproceedings{zhang2021computing,
  title={Computing ex ante coordinated team-maxmin equilibria in zero-sum multiplayer extensive-form games},
  author={Zhang, Youzhi and An, Bo and {\v{C}}ern{\`y}, Jakub},
  booktitle={AAAI},
  volume={35},
  pages={5813--5821},
  year={2021}
}

@inproceedings{farina2021connecting,
  title={Connecting optimal ex-ante collusion in teams to extensive-form correlation: {F}aster algorithms and positive complexity results},
  author={Farina, Gabriele and Celli, Andrea and Gatti, Nicola and Sandholm, Tuomas},
  booktitle={ICML},
  pages={3164--3173},
  year={2021} 
}

@inproceedings{sandholm2005mixed,
  title={Mixed-integer programming methods for finding {N}ash equilibria},
  author={Sandholm, Tuomas and Gilpin, Andrew and Conitzer, Vincent},
  booktitle={AAAI},
  pages={495--501},
  year={2005}
}

@inproceedings{conitzer2003complexity,
  title={Complexity results about {N}ash equilibria},
  author={Conitzer, Vincent and Sandholm, Tuomas},
  booktitle={IJCAI},
  pages={765--771},
  year={2003}
}

@inproceedings{zhang2022optimal,
  title={Optimal correlated equilibria in general-sum extensive-form games: {F}ixed-parameter algorithms, hardness, and two-sided column-generation},
  author={Zhang, Brian Hu and Farina, Gabriele and Celli, Andrea and Sandholm, Tuomas},
  booktitle={EC},
  pages={1119--1120},
  year={2022}
}

@inproceedings{carminati2022marriage,
  title={A marriage between adversarial team games and 2-player games: {E}nabling abstractions, no-regret learning, and subgame solving},
  author={Carminati, Luca and Cacciamani, Federico and Ciccone, Marco and Gatti, Nicola},
  booktitle={ICML},
  pages={2638--2657},
  year={2022},
  organization={PMLR}
}

@inproceedings{zhang2023team,
  title={Team Belief {DAG}: {G}eneralizing the Sequence Form to Team Games for Fast Computation of Correlated Team Max-Min Equilibria via Regret Minimization},
  author={Zhang, Brian Hu and Farina, Gabriele and Sandholm, Tuomas},
  booktitle={ICML},
  pages={40996-41018},
  year={2023}
}

@inproceedings{zhang2022teamtree,
  title={Team correlated equilibria in zero-sum extensive-form games via tree decompositions},
  author={Zhang, Brian Hu and Sandholm, Tuomas},
  booktitle={AAAI},
  pages={5252-5259},
  year={2022}
}

@article{lowe2017multi,
  title={Multi-agent actor-critic for mixed cooperative-competitive environments},
  author={Lowe, Ryan and Wu, Yi I and Tamar, Aviv and Harb, Jean and Pieter Abbeel, OpenAI and Mordatch, Igor},
  journal={NeurIPS},
  volume={30},
  year={2017}
}

@inproceedings{rashid2018qmix,
  title={{QMIX}: {M}onotonic Value Function Factorisation for Deep Multi-Agent Reinforcement Learning},
  author={Rashid, Tabish and Samvelyan, Mikayel and Schroeder, Christian and Farquhar, Gregory and Foerster, Jakob and Whiteson, Shimon},
  booktitle={ICML},
  pages={4295--4304},
  year={2018},
  organization={PMLR}
}

@inproceedings{zhang2022subgame,
  title={Subgame solving in adversarial team games},
  author={Zhang, Brian and Carminati, Luca and Cacciamani, Federico and Farina, Gabriele and Olivieri, Pierriccardo and Gatti, Nicola and Sandholm, Tuomas},
  booktitle={NeurIPS},
  pages={26686--26697},
  year={2022}
}

@inproceedings{Anagnostides2023Algorithms,
author = {Anagnostides, Ioannis and Kalogiannis, Fivos and Panageas, Ioannis and Vlatakis-Gkaragkounis, Emmanouil-Vasileios and Mcaleer, Stephen},
title = {Algorithms and Complexity for Computing {N}ash Equilibria in Adversarial Team Games},
booktitle={EC},
year = {2023}
}

@inproceedings{
mcaleer2023teampsro,
title={Team-{PSRO} for Learning Approximate {TMEC}or in Large Team Games via Cooperative Reinforcement Learning},
author={Stephen Marcus McAleer and Gabriele Farina and Gaoyue Zhou and Mingzhi Wang and Yaodong Yang and Tuomas Sandholm},
booktitle={NeurIPS},
year={2023} 
}

@article{aumann1974subjectivity,
  title={Subjectivity and correlation in randomized strategies},
  author={Aumann, Robert J},
  journal={Journal of mathematical Economics},
  volume={1},
  number={1},
  pages={67--96},
  year={1974},
  publisher={Elsevier}
}

@inproceedings{zhang2022correlation,
  title={Correlation-based algorithm for team-maxmin equilibrium in multiplayer extensive-form games},
  author={Zhang, Youzhi and An, Bo and Subrahmanian, V.S.},
  booktitle={IJCAI},
  pages={606--612},
  year={2022} 
}

@inproceedings{zhang2017optimal,
  title={Optimal escape interdiction on transportation networks},
  author={Zhang, Youzhi and An, Bo and Tran-Thanh, Long and Wang, Zhen and Gan, Jiarui and Jennings, Nicholas R},
  booktitle={IJCAI},
  pages={3936--3944},
  year={2017}
}

@inproceedings{li2021cfr,
  title={{CFR-MIX}: {S}olving imperfect information extensive-form games with combinatorial action space},
  author={Li, Shuxin and Youzhi Zhang and Wang, Xinrun and Xue, Wanqi and An, Bo},
   booktitle={IJCAI},
   pages={3663--3669},
  year={2021}
}

@inproceedings{xue2021solving,
  title={Solving large-scale extensive-form network security games via neural fictitious self-play},
  author={Xue, Wanqi and bf Youzhi Zhang and Li, Shuxin and Wang, Xinrun and An, Bo and Yeo, Chai Kiat},
    booktitle={IJCAI},
    pages={3713-3720},
  year={2021}
}

@inproceedings{li2023solving,
  title={Solving Large-Scale Pursuit-Evasion Games Using Pre-trained Strategies},
  author={Li, Shuxin and Wang, Xinrun and Zhang, Youzhi and Xue, Wanqi and {\v{C}}ern{\`y}, Jakub and An, Bo},
  booktitle={AAAI},
  volume={37},
  number={10},
  pages={11586--11594},
  year={2023}
}

@article{moulin1978strategically,
  title={Strategically zero-sum games: {T}he class of games whose completely mixed equilibria cannot be improved upon},
  author={ Moulin, H.  and  Vial, J. P. },
  journal={International Journal of Game Theory},
  volume={7},
  number={3-4},
  pages={201-221},
  year={1978},
}

@inproceedings{chen2006Settling,
  title={Settling the Complexity of Two-Player {N}ash Equilibrium},
  author={ Chen, Xi  and  Deng, Xiaotie },
  booktitle={FOCS},
  pages={261-272},
  year={2006},
}

@misc{Guardian2022,
  title={‘We didn’t accept it’: {DRC} minister laments forcing through of {COP}15 deal},
  author={Guardian},
  year={2022},
  howpublished={https://www.theguardian.com/environment/2022/dec/19/we-didnt-accept-it-drc-minister-laments-forcing-through-of-cop15-deal-aoe}
}

@misc{CNN2023,
  title={`Verge of complete failure’: Climate summit draft drops the mention of fossil fuel phase-out, angering advocates},
  author={CNN},
  year={2023},
  howpublished={https://edition.cnn.com/2023/12/11/climate/cop-28-draft-agreement-fossil-fuel-monday/index.html}
}

@misc{UN2023,
  title={},
  author={UN},
  year={2015},
  howpublished={https://sdgs.un.org/goals}
}

@inproceedings{sun2005coopetitive,
  title={Coopetitive game, equilibrium and their applications},
  author={Sun, Lihui and Xu, Xiaowen},
  booktitle={International Conference on Algorithmic Applications in Management},
  pages={104--111},
  year={2005},
  organization={Springer}
}

@article{okura2014coopetition,
  title={Coopetition and game theory.},
  author={Okura, Mahito and CARFI, David},
  journal={Journal of Applied Economic Sciences},
  volume={9},
  number={3},
  year={2014}
}

@article{bouncken2015coopetition,
  title={Coopetition: {A} systematic review, synthesis, and future research directions},
  author={Bouncken, Ricarda B and Gast, Johanna and Kraus, Sascha and Bogers, Marcel},
  journal={Review of managerial science},
  volume={9},
  pages={577--601},
  year={2015},
  publisher={Springer}
}

@article{bengtsson2000coopetition,
  title={``Coopetition” in business Networks—to cooperate and compete simultaneously},
  author={Bengtsson, Maria and Kock, S{\"o}ren},
  journal={Industrial marketing management},
  volume={29},
  number={5},
  pages={411--426},
  year={2000},
  publisher={Elsevier}
}

@book{brandenburger1996co,
  title={Co-opetition},
  author={Brandenburger, Adam M and Nalebuff, Barry J},
  year={1996},
  publisher={Crown Currency}
}

@article{daskalakis2009complexity,
  title={The complexity of computing a {N}ash equilibrium},
  author={Daskalakis, Constantinos and Goldberg, Paul W and Papadimitriou, Christos H},
  journal={Communications of the ACM},
  volume={52},
  number={2},
  pages={89--97},
  year={2009},
  publisher={ACM New York, NY, USA}
}

@article{gilboa1989nash,
  title={Nash and correlated equilibria: {S}ome complexity considerations},
  author={Gilboa, Itzhak and Zemel, Eitan},
  journal={Games and Economic Behavior},
  volume={1},
  number={1},
  pages={80--93},
  year={1989},
  publisher={Elsevier}
}

@article{rosen1965existence,
  title={Existence and uniqueness of equilibrium points for concave n-person games},
  author={Rosen, J Ben},
  journal={Econometrica: Journal of the Econometric Society},
  pages={520--534},
  year={1965},
  publisher={JSTOR}
}

@article{marschak1955elements, title={Elements for a theory of teams}, author={Marschak, Jakob}, journal={Management science}, volume={1}, number={2}, pages={127--137}, year={1955}, publisher={INFORMS} }

@misc{wiki2025ce,
  title={Correlated equilibrium},
  author={Wiki},
  year={2025},
  howpublished={\url{https://en.wikipedia.org/wiki/Correlated_equilibrium}}
}

@inproceedings{chan2012interdependent, title={Interdependent defense games: modeling interdependent security under deliberate attacks}, author={Chan, Hau and Ceyko, Michael and Ortiz, Luis E}, booktitle={UAI}, pages={152--162}, year={2012} }

@article{hu2003nash, title={Nash {Q}-learning for general-sum stochastic games}, author={Hu, Junling and Wellman, Michael P}, journal={Journal of machine learning research}, volume={4}, number={Nov}, pages={1039--1069}, year={2003} }

@article{zhong2024heterogeneous, title={Heterogeneous-agent reinforcement learning}, author={Zhong, Yifan and Kuba, Jakub Grudzien and Feng, Xidong and Hu, Siyi and Ji, Jiaming and Yang, Yaodong}, journal={Journal of Machine Learning Research}, volume={25}, number={32}, pages={1--67}, year={2024} }

@article{yu2022surprising, title={The surprising effectiveness of {PPO} in cooperative multi-agent games}, author={Yu, Chao and Velu, Akash and Vinitsky, Eugene and Gao, Jiaxuan and Wang, Yu and Bayen, Alexandre and Wu, Yi}, journal={NeurIPS}, volume={35}, pages={24611--24624}, year={2022} }

@inproceedings{song2023multi, title={Multi-defender Security Games with Schedules}, author={Song, Zimeng and Ling, Chun Kai and Fang, Fei}, booktitle={GameSec}, pages={65--85}, year={2023}, organization={Springer} }


\end{document}